\documentclass[epj]{svjour}

\usepackage{graphics}
\usepackage{dsfont}
\usepackage{amsfonts}
\usepackage{amsmath}
\usepackage{booktabs}

\newcommand{\lag}{\tau}

\newcommand{\myfactside}{0.5}
\newcommand{\myspaceside}{-0.1cm}

\newcommand{\myspace}{-0.4cm}

\newcommand{\meanlambda}{\bar\lambda}
\newtheorem{notations}{Notations}
\newtheorem{cor}{Corollary}
\newtheorem{lem}{Lemma}

\newcommand{\otoprule}{\midrule[\heavyrulewidth]}

\begin{document}


\title{Non-parametric kernel estimation for symmetric Hawkes processes. Application to high frequency financial data.}
\titlerunning{Non-parametric Hawkes kernel estimation}
\author{Emmanuel Bacry \inst{1} \and Khalil Dayri \inst{1} \and Jean-Francois Muzy \inst{1,2}
\thanks{This research is part of the Chair {\it Financial Risks} of the {\it Risk Foundation.}}
\thanks{The financial data used in this paper  have been  provided by the company {\em QuantHouse
EUROPE/ASIA}, http://www.quanthouse.com}.
}
\institute{CMAP UMR 7641 CNRS, Ecole Polytechnique, 91128 Palaiseau, France \and SPE UMR 6134 CNRS, Universit\'e de Corse, 20250 Corte, France}

\date{\today}

\abstract{
We define a numerical method that provides a non-parametric estimation of
the kernel shape in symmetric multivariate Hawkes processes. This method relies
on second order statistical properties of Hawkes processes that relate the
covariance matrix of the process to the kernel matrix. The square root
of the correlation function is computed using a minimal phase recovering method.
We illustrate our method on some examples and provide an empirical study
of the estimation errors. Within this framework, we analyze high frequency financial price
data modeled as 1D or 2D Hawkes processes. We find slowly decaying (power-law) kernel shapes suggesting
a long memory nature of self-excitation phenomena at the microstructure level of price dynamics.
}

\maketitle

\section{Introduction}
Although the concept of self excitement was commonly used for a long time by seismologists (see \cite{vere-jones_stochastic_1970} and the references therein), Alan Hawkes was the first to provide a well defined point process with a self exciting behavior \cite{hawkes_point_1971,hawkes_spectra_1971,liniger_multivariate_2009}.
This model was introduced to reproduce the ripple effects generated after the occurrence of an earthquake \cite{vere-jones_examples_1982,adamopoulos_cluster_1976}. Since then, it has been successfully used in many areas
ranging from seismology (see e.g., \cite{ogata_seismicity_1999}, for a recent review), biology \cite{reynaud-bouret_adaptive_2010}, to even criminology and terrorism \cite{mohler_self-exciting_2011,e._lewis_self-exciting_2011} (cf \cite{liniger_multivariate_2009} and references therein for a detailed review of Hawkes processes and their applications). As far as financial applications are concerned, since transactions and price changes are discrete events, Hawkes processes have naturally been generating more and more interest.
Applications can be found in the field of order arrival rate modeling \cite{hewlett_clustering_2006,bowsher_modelling_2007,Ioane2011Hawkes}, noise microstructure dynamics \cite{bacry_modeling_2011}, volatility clustering \cite{embrechts_multivariate_2011}, extreme values and VaR estimation \cite{chavez-demoulin_estimating_2005} and credit modeling \cite{giesecke_top_2007}.

Hawkes models account for a self exciting behavior of events
by which the arrival of one event increases the probability of occurrence of new ones. In its most basic form and in the one dimensional case, the Hawkes process is a counting process defined by $\lambda_t$, the rate of arrival of events by:
\begin{equation}\label{eq:HawkesBasicModelDef}
    \lambda_t=\mu+ \int_{-\infty}^t\phi_{t-s}dN_s
\end{equation}
where $\mu>0$ is a constant \emph{background} rate, $N_t$ is the cumulative counting process and $\phi$ a positive real function called \emph{decay kernel}.
We can clearly see in equation (\ref{eq:HawkesBasicModelDef}) that when some event occurs at time $t$, we have $dN_t=1$ and hence $d\lambda_t=\phi_0$. The influence of the event is transmitted to future times through $\phi$ such that at time $u>t$ the increase in $\lambda_t$ due to the time $t$ event is $\phi_{u-t}$. Thus a self exciting behavior is observed.

A basic issue in applications of Hawkes processes concerns their estimation. In early applications, parametric estimation used spectral analysis by means of the Bartlett spectrum (the Fourier transform of the autocovariance of the process) \cite{bartlett_spectral_1963,bartlett_spectral_1964} and it was indeed through that light that Hawkes presented his model. A maximum likelihood method for estimating the parameters of exponential, power law and Laguerre class of decay kernels was developed in \cite{ozaki_maximum_1979,ogata_linear_1982} and it became the standard method for estimating Hawkes processes. Furthermore, the Laguerre decay kernels were seen to be very pertinent decay kernels because they allowed to account for long term dependencies as well as offering short term flexibility. More recently, other types of estimation procedures were developed. When the form of the decay kernel is unknown, non-parametric methods are desirable because they give an idea of their general shape. For example, by using the branching property of the Hawkes process \cite{zhuang_stochastic_2002}, the authors in \cite{marsan_extending_2008} and \cite{lewis_nonparametric_2011} were able to provide Expectation-Maximization algorithms to estimate both background rate and the decay kernel.
In \cite{reynaud-bouret_adaptive_2010}, the authors present also an algorithmic method for decay estimation by using a penalized projection method.

In this paper, we propose an alternative simple non parametric estimation method for multivariate symmetric Hawkes processes based on the Bartlett spectrum. We present the method and its numerical feasibility without going too much in details about convergence speeds or error optimization. By studying one dimensional and 2-dimensional examples, we show that our approach provides reliable estimates on both fast and slowly decaying kernels. We then apply our method to high frequency financial data and find power-law kernels. This implies that the arrival of events display long range correlations, a phenomenon well known in finance, similarly to what was suggested in \cite{mikosch_modelling_2009}.

This paper is organized as follows. In section \ref{Sec:MultivariateHawkesProcesses}, we introduce a basic version of a multivariate Hawkes processes. We place ourselves in the context of an $n-$dimentional linear Hawkes process with a constant background rate and a nonnegative decay kernel. We set some notations and give out a martingale representation of the rate function $\lambda_t$ in Eq. (\ref{eq:HawkesBasicModelDef}). In section \ref{Sec:Cov}, we study the autocovariance of the process. Along the same line as \cite{hawkes_spectra_1971}, we establish its relationship with the decay kernel in both direct and Fourier spaces. The estimation method in the case of symmetric Hawkes processes is then provided in section \ref{sec:Estimation}. This method, based on a Hilbert transform phase recuperation method, is explicitly detailed in both univariate and special symmetric bivariate cases. In section \ref{Sec:Num}, the method is illustrated by numerical examples for both exponential and power law kernels. We also address some statistical issues concerning the estimation errors. Finally, in section \ref{sec:Applica-high-frequen-market-data}, we apply our method to high frequency financial data for which we deduce a long range nature of the decay kernels.

\section{Multivariate Hawkes Processes \label{Sec:MultivariateHawkesProcesses}}
\subsection{Notations and Definitions}
\label{Sec:MultivariateHawkesProcessesDef}
As introduced by Hawkes in \cite{hawkes_point_1971} and \cite{hawkes_spectra_1971}, let us consider an n-dimensional
point process $(N_t)_{t\geq 0}$, where $N_t^i$ $1\leq i \leq n$
represents the cumulative number of events in the $i^{th}$ component of the process $N_t$ up to time $t$.

The conditional intensity vector at any time $t$ is assumed to be a random process that depends on the past as
\begin{equation}\label{eq:Lambda_EDS_N}
    \lambda_t=\mu+ \int_{-\infty}^t\phi_{t-s}dN_s
\end{equation}
where $\mu$ is a vector of size $n$ with strictly positive components ($\mu^i>0$) and $\phi_t$  is an $n\times n$ matrix
 referred to as the \emph{decay kernel}. \\

\begin{notations}
\label{not1}
In the following
\begin{itemize}
\item $\mathcal{M}_{n,p}(\mathds{R})$ (resp. $\mathcal{M}_{n,p}(\mathds{C})$) denotes the set of $n\times p$ matrices with values in $\mathds{R}$ (resp. $\mathds{C}$). For any matrix $M$ (resp. vector $v$), $M^{ij}$ (resp. $v^i$) denotes its elements.
\item For any function $f_t$, $\widehat f_z = \int_{\mathds{R}} e^{-zt}f_t dt$ corresponds to its Laplace transform.
\item By extension if $M_t
\in \mathcal{M}_{n,p}(\mathds{R})$ then $\widehat M_z \in \mathcal{M}_{n,p}(\mathds{C})$ corresponds to the matrix whose elements are the Laplace transforms of the elements of $M_t$.
\end{itemize}
\end{notations}


Using these notations, one of the main results of Hawkes  is that if

\begin{itemize}
\item[{\bf H1}]  the kernel $\phi_t \in \mathcal{M}_{n,n}$  is positive ($\phi_t^{ij}\ge 0$, $\forall t$) and causal
($\phi_t = 0$, $\forall t<0$),
\item[ and]
\item[\bf H2] the spectral radius of $\widehat {\phi}_{0}$ (i.e., its largest eigen value) is strictly smaller than 1,
\end{itemize}
\noindent
then $(N_t)_{t\ge 0}$ is a n-dimensional point process with stationary increments. The conditional intensity $\lambda_t$ is itself a stationary process with mean
\begin{equation}
\label{eq:Lambda_unconditional}
\Lambda=E(\lambda_t) = E(dN_t)/dt.
\end{equation}
Combining this last equation with Eq. (\ref{eq:Lambda_EDS_N}), one easily gets $\Lambda =\mu+ \int_{-\infty}^t\phi_{t-s}ds \Lambda =\mu+ \int_0^{\infty}\phi_udu \Lambda$ and consequently
\begin{equation}
\Lambda =(\mathbb{I}- \widehat \phi_0)^{-1}\mu,
\end{equation}
where $\mathbb{I}$ refers to the $n\times n$ identity matrix.

Before moving on, we need to introduce some more notations that will be used all along the paper.
\begin{notations} \label{not2} If $A_t\in\mathcal{M}_{m,n}(\mathds{R})$ and $B_t\in\mathcal{M}_{n,p}(\mathds{R})$ then the convolution product of $A_t$ and $B_t$ is naturally defined as $A\star B_t =\int_{\mathds{R}} A_sB_{t-s}ds=\int_{\mathds{R}} A_{t-s}B_sds$. Of course it is associative and distributive however it is generally not commutative (unless $A_t$ and $B_t$ are commutative). The neutral element is $\delta \mathbb{I}_t$, i.e., the diagonal matrix with Dirac distribution on the diagonal.
 In the following we will use the notation :
\begin{equation}
\nonumber
A \star dN_t=\int_{\mathds{R}}A_{t-s}dN_s.
\end{equation}
\end{notations}
Combining both Notations \ref{not1} and \ref{not2}, it is easy to show that
the convolution theorem on matrices translates in
\begin{equation}
\nonumber
\widehat{A \star B}_z = \widehat A_z \widehat B_z.
\end{equation}


\subsection{Martingale representation of $\lambda_t$}
We roughly follow a similar path to Hawkes in \cite{hawkes_spectra_1971}.
Let $(M_t)_{t\geq 0}$ be the martingale compensated process of $(N_t)_{t\geq 0}$ defined by:
\begin{equation}\label{eq:Relation_N_to_M}
    dM_t=dN_t - \lambda_tdt \; .
\end{equation}
Then $\lambda_t$ can be represented as a stochastic integral with respect to the martingale $(M_t)_{t\geq 0}$:
\begin{proposition}
One has:
\begin{equation}\label{eq:Lambda_EDS_M}
    \lambda_t=\Lambda + \Psi\star dM_t,
\end{equation}
where $\Psi_t$ is defined as
\begin{equation}
\label{eq:Psi}
\Psi_t= \sum_{n=1}^{\infty}\phi^{(\star n)}_t,
\end{equation}
where $\phi^{(\star n)}_t$ refers to the $n^{th}$ auto-convolution of $\phi_t$ (i.e.,  $\widehat{\phi^{(\star n)}}_z = (\widehat \phi_z)^{n}$)
\end{proposition}
\begin{proof}
Using equations (\ref{eq:Lambda_EDS_N}) and (\ref{eq:Relation_N_to_M})  one has:
\begin{equation}
\nonumber
    \lambda_t = \mu+ \phi\star dN_t=\mu+ \phi\star dM_t + \phi\star \lambda_t, \\
\end{equation}
and consequently
\begin{equation}
\nonumber
    (\delta\mathbb{I} - \phi)\star \lambda_t = \mu+ \phi\star dM_t.
\end{equation}
Let us note that the inverse of $ \delta\mathbb{I}_t - \phi_t$ for the convolution product is nothing but $h_t = \delta \mathbb{I}_t + \Psi_t$. Thus
convoluting on each side of the last equation by $h_t$, one gets (since $h\star\phi_t = \Psi_t$)
\begin{equation}
\nonumber
   \lambda_t=h\star\mu+ \Psi \star dM_t.
\end{equation}
Since $\mu$ is a constant one has $h\star\mu=  \mu \widehat h_0$. Using the convolution theorem one gets $\mu  \widehat h_0= (\mathbb{I}-\widehat {\phi}_0)^{-1} \mu= \Lambda$ which proves the proposition.
\end{proof}

\section{The covariance matrix of Hawkes processes \label{Sec:Cov}}
The kernel estimator we are going to build is based on the empirical auto-covariance of $(N_t)_{t\ge 0}$. This section is devoted to the covariance matrix of the n-dimensional Hawkes processes. However, we first discuss some useful results about their  {\em infinitesimal} auto-covariance function.

\subsection{The infinitesimal covariance}

Let us define the infinitesimal covariance matrix:
\begin{equation*}
\nu_{t-t'}=E(dN_tdN_{t'}^\dagger),
\end{equation*}
where $M^\dagger$ denotes the hermitian conjugate of a matrix $M$. Along the same way as Hawkes in \cite{hawkes_spectra_1971}, $\nu_{t-t'}$ can be related to the decay kernel $\phi$. We present in proposition \ref{prop:InfinitesimalCov} an equation linking the infinitesimal covariance matrix to $\Lambda$ and $\Psi$ but, unlike Hawkes, we express $\nu$ explicitly as a function of $\Psi$ and $\Lambda$ instead of an integral equation in $\nu$. This result will be at the heart of the estimation method we propose in this paper.

\begin{proposition}\label{prop:InfinitesimalCov}
Let $(N_t)_{t \ge 0}$ be an $n$-dimensional Hawkes process with intensity $\lambda_t$ as defined in Section \ref{Sec:MultivariateHawkesProcessesDef} (assuming both \textbf{H1} and \textbf{H2}). Let  $\widetilde{\Psi}_t=\Psi_{-t}$. We have the following result:
\begin{equation}
\begin{split}
    E(dN_tdN_{t'}^\dagger) = & \left(\Lambda\Lambda^\dagger +  \Sigma\delta_{t-t'} + \Psi_{t-t'} \Sigma + \Sigma\Psi^\dagger_{t'-t} \right. \\
      & \left. + \; \widetilde{\Psi}\star\Sigma\Psi^\dagger_{t'-t} \right) dtdt'
\end{split}
\label{eq:EdNdNPrime}
\end{equation}
where $\Sigma$ is the diagonal matrix defined by $\Sigma^{ii} = \Lambda^i$ for all $1\leq i \leq n$ and $\delta_t$ is the Dirac distribution.
\end{proposition}
\begin{proof}
By using equation (\ref{eq:Relation_N_to_M}), we can write:
\begin{eqnarray}
  E(dN_tdN_{t'}^\dagger) &=& E((dM_t + \lambda_tdt)(dM_{t'} + \lambda_{t'}dt')^\dagger) \nonumber \\
   &=&  E(dM_tdM_{t'}^\dagger) \label{eq:Sub:dNdNPrime1}\\
   && + E(\lambda_tdM_{t'}^\dagger) \label{eq:Sub:dNdNPrime2}dt\\
   && + E(dM_t  \lambda_{t'}^\dagger) \label{eq:Sub:dNdNPrime3} dt' \\
   && + E(\lambda_t \lambda_{t'}^\dagger) \label{eq:Sub:dNdNPrime4}dtdt'
\end{eqnarray}
We begin by noticing that (thanks to martingale property), $E(dM_tdM_{t'}^\dagger) =0,~\forall t\neq t'$.
As for when $t=t'$, we have for all $1\leq i < j \leq n$
\begin{equation*}
E(dM^i_tdM_t^j)=0
\end{equation*}
since the $N^i$'s (and hence the $M^i$'s) for $1\leq i \leq n$ have no jump in common. Moreover, for $i=j$ we have
\begin{equation*}
E(dM^i_tdM_t^i)=\Lambda^{i}dt
\end{equation*}
because $E(dM^i_tdM_t^i)=E(dN^i_tdN_t^i)$ and $E(dN^i_tdN_t^i)=E(dN^i_t)=\Lambda^{i}dt$ since the jumps of $N_t$ are of size $1$.
\noindent To sum up, the term (\ref{eq:Sub:dNdNPrime1}) becomes:
\begin{equation}\label{eq:Subart1}
    E(dM_tdM_{t'}^\dagger)=\Sigma\delta_{t-t'}dtdt'
\end{equation}

\noindent The remaining  terms can be then calculated along the same line.
Replacing $\lambda_t$'s expression from equation (\ref{eq:Lambda_EDS_M}) in the term (\ref{eq:Sub:dNdNPrime2}),  gives:
\begin{equation}
\nonumber
  E(\lambda_t dM_{t'}^\dagger) dt= E\left(\Psi \star dM_t dM_{t'}^\dagger \right) dt,
\end{equation}
or in other words
\begin{equation}
\nonumber
E(\lambda_tdM_{t'}^\dagger)dt=\int_{\mathds{R}}\Psi_{t-s} E(dM_sdM_{t'}^\dagger) dtds
\end{equation}
and thanks to equation (\ref{eq:Subart1}),
\begin{equation}
\nonumber
E(\lambda_tdM_{t'}^\dagger)dt= \int_{\mathds{R}}\Psi_{t-s} \Sigma\delta_{s-t'}dsdt' dt
\end{equation}
that is simplified to:
\begin{equation}\label{eq:Subart2}
E(\lambda_t dM_{t'}^\dagger)dt   = \Psi_{t-t'}\Sigma dtdt'
\end{equation}
Similarly, the term (\ref{eq:Sub:dNdNPrime3}) becomes:
\begin{equation}
    E(dM_t \lambda_{t'}^\dagger)dt' =\Sigma\Psi^\dagger_{t'-t}dtdt'
\end{equation}
and finally, using Eq. (\ref{eq:Subart1}), the term (\ref{eq:Sub:dNdNPrime4}) can be written as:
\begin{eqnarray*}
E(\lambda_t \lambda_{t'}^\dagger) dtdt'   &=& \left(\Lambda E(\lambda_{t'}^\dagger) + E((\Psi\star dM_t) \lambda^\dagger_{t'}\right)dtdt' \\
                                &=& \left(\Lambda \Lambda^\dagger + \int_{\mathds{R}}\Psi_{t-s}E(dM_s\lambda_{t'}^\dagger ds)\right) dtdt'\\
                                &=& \left(\Lambda \Lambda^\dagger + \int_{\mathds{R}}\Psi_{t-s}\Sigma \Psi^\dagger_{t'-s} ds)\right) dtdt'
\end{eqnarray*}
By setting $u=t'-s$, we get:
\begin{equation}
    E(\lambda_t \lambda_{t'}^\dagger)dtdt' = \left(\Lambda \Lambda^\dagger + \widetilde{\Psi} \star \Sigma \Psi^\dagger_{t'-t}\right) dtdt',
\end{equation}
which ends the proof of the proposition.
\end{proof}

\subsection{The covariance matrix}\label{Sec:SampledACV}
The (normalized) covariance matrix of the Hawkes process can be defined, at scale $h$ and lag $\lag$, by
\begin{equation}\label{eq:SampledACV}
v^{(h)}_{\lag} =  h^{-1} \mbox{Cov} \left(N_{t+h}-N_t,N_{t+h+\lag}-N_{t+\lag}\right) \; ,
\end{equation}
where we normalized by $h$ in order to avoid a trivial scale dependence. Let us note that, since the increments of $N_t$ are stationary, the previous definition does not depend on $t$. Thus, it can be rewritten as
\begin{equation}\label{eq:SampledACV1}
v^{(h)}_{\lag} =  \frac{1}{h}E\left((\int_0^{h}dN_s-\Lambda h)(\int_{\lag}^{\lag+h}dN_s^\dagger-\Lambda^\dagger h)\right)
\end{equation}
It is clear that this quantity can be easily estimated on real data using empirical means. As we will see, the non parametric estimation of the decay kernel we propose in this paper is based on these empirical estimations. More precisely, it is based on the following Theorem.
\begin{theorem}
Let $g^{(h)}_t=(1-\frac{|t|}{h})^+$. $v^{(h)}_{\lag} $ can be expressed as a function of $g^{(h)}_{\lag}$ and $\Psi_{\lag}$:
\begin{equation}\label{eq:SignaturePlotExpression}
    v^{(h)}_{\lag} =g^{(h)}_{\lag}\Sigma + g^{(h)}\star\Psi_{-\lag}\Sigma + g^{(h)}\star\Sigma\Psi^\dagger_{\lag} + g^{(h)}\star\widetilde{\Psi}\star\Sigma\Psi^\dagger_{\lag}
\end{equation}
\end{theorem}\label{thm:SignaturePlotExpression}

\begin{proof}
Let us begin the proof by noticing that for any function $f$ with values in $\mathds{R}^+$ we have:
\begin{equation}\label{eq:Lemm1}
  h^{-1}  \int_0^{h}\int_{\lag}^{\lag+h}f_{t-t'}dt'dt= f \star g^{(h)}_{-\lag}
\end{equation}
It follows that
\begin{eqnarray*}
  v^{(h)}_{\lag} &=& \frac{1}{h}E\left(\int_0^{h}dN_s\int_{\lag}^{\lag+h}dN_s^\dagger \right.\\
                 &-& \left. \int_0^{h}dN_s\Lambda^\dagger h - \Lambda h \int_{\lag}^{\lag+h}dN_s^\dagger + \Lambda \Lambda^\dagger h^2\right)\\
    &=& \frac{1}{h}E\left(\int_0^{h}\int_{\lag}^{\lag+h}dN_tdN_{t'}^\dagger -\Lambda \Lambda^\dagger h^2\right) \; .
\end{eqnarray*}
Using equation (\ref{eq:EdNdNPrime}), we can split $dN_tdN_{t'}^\dagger$ into four parts and write:
\begin{eqnarray*}
  v^{(h)}_{\lag} &=&  \frac{1}{h}\int_0^{h}\int_{\lag}^{\lag+h}\left(\Sigma\delta_{t-t'} \right. \\
                 &+& \left. \Psi_{t-t'}\Sigma + \Sigma\Psi^\dagger_{t'-t} + \widetilde{\Psi}\star\Sigma\Psi^\dagger_{t'-t}\right) dtdt'\nonumber
\end{eqnarray*}
By applying equation (\ref{eq:Lemm1}) to each of the terms under the double integral we get equation (\ref{eq:SignaturePlotExpression}) and
achieve the proof.
\end{proof}

It is more convenient to rewrite the result of theorem \ref{thm:SignaturePlotExpression} in Laplace domain.
Since $\widehat{\phi^{(\star n)}}_z = {\widehat \phi^n}_z$, Eq. (\ref{eq:Psi}) translates into :
\begin{equation}
\label{eq:PsiHatFuncFiHat}
\widehat\Psi_z = \sum_{n=1}^{+\infty} {\widehat \phi}^n_z = \widehat \phi_z(\mathbb{I}-\widehat \phi_z)^{-1}
\end{equation}
and conversely:
\begin{equation}\label{eq:FiHatFuncPsiHat}
\widehat \phi_z=(\mathbb{I}+\widehat \Psi_z)^{-1}\widehat{\Psi_z}
\end{equation}
Theorem \ref{thm:SignaturePlotExpression} gives way to the following corollary.
\begin{cor}
In Laplace domain equation (\ref{eq:SignaturePlotExpression}) becomes:
\begin{equation}\label{eq:SignaturePlotExpressionFourier}
    \widehat v^{(h)}_z =\widehat g^{(h)}_z(\mathbb{I} + \widehat \Psi^*_z)\Sigma(\mathbb{I} + \widehat \Psi^*_z)^\dagger
\end{equation}
\end{cor}

\section{Non-parametric estimation of the kernel $\phi_t$}\label{sec:Estimation}
\subsection{The estimation principle}
\label{estimation_principle}
In this paper, we aim at building an estimator of $\phi_t$ based on  empirical measurements of $v^{(h)}_\lag$. Let us note that empirical measurements of $v^{(h)}_\lag$ are naturally obtained replacing probabilistic mean by empirical mean in Eq. (\ref{eq:SampledACV1}) (see \cite{bacry_scaling_2011} for proof of convergence of the empirical mean towards the probabilistic mean). Thus, in order to build an estimator, we need to express $\phi_t$  as a function of $v^{(h)}_\lag$. In the Laplace domain, since a given $\widehat \psi_z$ corresponds to a unique $\widehat \phi_z$ (through Eq. (\ref{eq:FiHatFuncPsiHat})), it translates in trying to  express $\widehat \psi_z$ as a function of $\widehat v^{(h)}_{z}$ which exactly corresponds to inverting Eq. (\ref{eq:SignaturePlotExpressionFourier}) (i.e., computing the square root of $(\mathbb{I}+\widehat \Psi^*_z)\Sigma(\mathbb{I} + \widehat \Psi^*_z)^\dagger$). Indeed, knowing $\widehat \psi_z$, one easily gets $\widehat {\phi}_z$ (from Eq. \eqref{eq:FiHatFuncPsiHat}) and finally $\phi_t$. Actually, let us note that, from a practical point of view, we don't need to work in the full complex domain $z \in \mathds{C}$ of the Laplace transform. Working with the Fourier transform restriction (i.e., $z = i\omega$ with $\omega \in \mathds{R}$) is enough to recover $\phi_t$.

\paragraph{Dealing with cancelations of $\widehat g^{(h)}_z$.}

The first problem that seems to appear for inverting this formula (\ref{eq:SignaturePlotExpressionFourier}) (i.e., expressing $\widehat \psi_z$ as a function of $\widehat v^{(h)}_{z}$ only for $z = i\omega$) is the scalar term $\widehat g^{(h)}_z$ that may vanish. Indeed, since $g^{(h)}_t=(1-\frac{|t|}{h})^+$, its Fourier transform, $\widehat g^{(h)}_{i\omega} = (4/\omega^2h)\sin^2(\omega h/2)$, cancels for all $\omega$ of the form $\frac{2n\pi}{h}$, $n \in \mathds{Z}, \, n\neq 0$. Actually, this is not a real problem as long as
\begin{itemize}
\item $\lag$ in the empirical estimation of $v^{(h)}_\lag$ is sampled using a sampling period $\Delta$ greater than $h/2$
\item $\Delta$ is small enough so that $\widehat v^{(h)}_{i \omega}$ can be considered to have a compact support
\end{itemize}
Indeed, if this is the case, then the estimation of $\widehat v^{(h)}_{i\omega}$ (practically obtained by taking the Discrete Fourier Transform (DFT) of the sampled signal $(v^{(h)}_{k\Delta})_k)$
will be equal on $[-\pi/\Delta,\pi/\Delta]$ to the product of the DFT of $g^{(h)}_\lag$ (which does not vanish since
$[-\pi/\Delta,\pi/\Delta] \subset [-\pi/h,\pi/h]$)
and the DFT of $(\delta_\lag+\widetilde{\Psi})\Sigma\star(\delta_\lag+\widetilde\Psi)^\dagger_{\lag}$.
Consequently, as long as $h$ is small enough (so that $\Delta$ can be chosen greater than $h$ and small enough), dividing on both hand sides Eq. (\ref{eq:SignaturePlotExpressionFourier}) by $\widehat g^{(h)}_{z}$ is not, from a practical estimation point of view, a real problem. So, in the following we will write
\begin{equation}\label{eq:SignaturePlotExpressionFourier1}
   (\mathbb{I} + \widehat \Psi^*_z)\Sigma (\mathbb{I} + \widehat \Psi^*_z)^\dagger =  \widehat v_z^{(h)}/\widehat g^{(h)}_z,
\end{equation}
without bothering with eventual cancelations of $\widehat g^{(h)}_z$.

\paragraph{Computing the square root of $(\mathbb{I}+\widehat \Psi^*_z)\Sigma(\mathbb{I} + \widehat \Psi^*_z)^\dagger$.}
In the estimation process, once we have estimated  $v^{(h)}_{\lag}$ and consequently  (through DFT) $\widehat v^{(h)}_{z}$ (for $z=i\omega$), using Eq \eqref{eq:SignaturePlotExpressionFourier1}, we can estimate $ (\mathbb{I} + \widehat \Psi^*_z) \Sigma (\mathbb{I} + \widehat \Psi^*_z)^\dagger$. We need to go from there to the estimation of $\widehat \Psi_z$ (then using Eq. (\ref{eq:FiHatFuncPsiHat}), we get $\widehat {\phi}_z$ and then by inverse Fourier transform ${\phi}_t$). This problem requires therefore to take the square root of the left hand side of Eq. (\ref{eq:SignaturePlotExpressionFourier1}).
In dimension $n=1$, it means being able to go from $|1 + \widehat \Psi_z|^2$ to $\widehat{\Psi}_z$. There is clearly a phase determination problem. We will see that, in dimension $n=1$, this phase is uniquely determined by the hypothesis \textbf{H1} and \textbf{H2}. However, in dimension $n>1$, they are many possible solutions and determining the correct one is not necessarily possible in general. We need to make a strong additional hypothesis.

\subsection{Further hypothesis on the kernel $\phi_t$}\label{subsec:assump}

In the following, we suppose that

\begin{itemize}
\item[{\bf H3}] $\Sigma = \meanlambda \mathbb{I}$ with  $\forall ~ i$, $E(\lambda_t^{i})=\meanlambda \in \mathds{R}$, and $\widehat \phi_z$ can be diagonalized into a matrix $D$ with some constant unitary matrix $U$ (that does not depend on $z$):
\begin{equation}
\label{eq:assumption}
\widehat \phi_z = U^{\dagger}D_z U.
\end{equation}
\end{itemize}

Let us point out that, though always true in dimension $n=1$,   \textbf{H3} is a strong hypothesis for dimension $n\ge 2$. Clearly it is satisfied in the case all the components of the process are identically distributed, i.e.\, the process is invariant under arbitrary permutations. Other cases are very specific and consequently are not discussed in this paper.

\subsection{The estimator}\label{ssec:estimat}
Using  Eq. \eqref{eq:PsiHatFuncFiHat} along with \textbf{H3}, one gets that $\mathbb{I} + \widehat \Psi^*_z$ is diagonalizable in the same basis $U$ and that
\begin{equation}
\mathbb{I} + \widehat \Psi^*_z = U^\dagger (\mathbb{I} -D_z^*)^{-1} U,
\end{equation}
and thus from (\ref{eq:SignaturePlotExpressionFourier1}), one gets
\begin{equation}
\label{eq:Ez}
 E_z = U \widehat v^{(h)}_{z} U^{\dagger}/(\meanlambda \widehat g^{(h)}_z),
\end{equation}
where $E_z$ is the diagonal matrix with the real positive coefficients :
\begin{equation}
E_z^{kk} = |1 - D_z^{kk}|^{-2}.
\end{equation}

So the estimation problem reduces to being able to recover the coefficients $D_z^{kk}$ from the coefficients $E_z^{kk} = |1 - D_z^{kk}|^{-2}$. This problem is solved by the following Lemma.

\begin{lem}
Let $k$ be fixed. Let $z = i\omega$ (with $\omega \in \mathds{R}$).
Then
\begin{equation}
\label{eq:lem}
(1-D^{(kk)}_{i\omega})^{-1} = e^{\frac 1 2 \log|E_{i\omega}^{kk}| -i H(\frac 1 2 \log|E_{i\omega}^{kk}|)},
\end{equation}
where the operator $H(.)$ refers to the Hilbert transform \cite{oppenheim_discrete-time_1999}.
\end{lem}

The proof is based on the following theorem:

\begin{theorem}[Paley-Wiener \cite{paley_fourier_1934}]
Let's suppose that we observe the amplitude $|\widehat f_{i\omega}|$ of the Fourier transform of a real filter $f_t$.
If
\begin{equation}\label{eq:Paley-Wiener}
    \int_{\mathds{R}}\frac{\log(|\widehat{f}_{i\omega}|)}{1+\omega^2}d\omega ~~ < ~~ \infty,
\end{equation}
then the filter $g_t$ defined by its Fourier transform
\begin{equation}
\widehat{g}_{i\omega}=e^{\log(|\widehat{f}_{i\omega}|) - i H(\log(|\widehat{f}_{i\omega}|))},
\end{equation}
is the only causal filter (i.e., supported by $\mathds{R}^+$), which is a phase minimal filter\footnote{a minimal phase filter \cite{oppenheim_discrete-time_1999} is a filter whose all the zeros and the poles of its Laplace transform satisfy $\Re(z)<0$} and which satisfies $|\widehat g_{i\omega}| = |\widehat f_{i\omega}|$.
\end{theorem}

\noindent
{\em Proof of the Lemma.}
It is a simple application of this theorem with $|\widehat f_{i\omega}| = \sqrt{E_z^{kk}} = |1 - D_{i\omega}^{(kk)}|^{-1}$. Indeed, let us first show that $\widehat g_{z} = (1 - D_{z}^{(kk)})^{-1}$ is a minimal phase filter, i.e.,  that both the poles and zeros are such that $\Re(z)<0$.
\begin{itemize}
\item Let $z$ be a zero of $(1 - D_z^{kk})^{-1}$, then it is a pole of $D_z^{kk}$ and consequently of $\widehat \phi_z$. However, from \textbf{H1} and \textbf{H2}
one concludes that
$|\widehat \phi_z| \le \int_0^{+\infty} |e^{-zt}| \phi_t dt$  cannot be infinite, unless $\Re(z)<0$.
\item Let $z$ be a pole of $(1 - D_z^{kk})^{-1}$. It thus satisfies $D_z^{kk} = 1$. Thus 1 is an eigenvalue of $\widehat {\phi}_{0}$ which is in contradiction with \textbf{H2}. Thus $(1 - D_z^{kk})^{-1}$ has no pole.
\end{itemize}
Consequently $\widehat g_{z} =(1 - D_z^{kk})^{-1}$ is a minimal phase filter.
Moreover, since every coefficients of $\phi_t$ are positive and in $L^1$, the Fourier transforms
$D_{i\omega}^{(kk)}$, for any $k$,  are continuous functions of $\omega$ and goes to 0 at infinity. Along with the fact that $(1 - D_z^{kk})^{-1}$ has no pole and has its zeros on the half-plane $\Re(z)<0$ we easily conclude that $|\widehat f_{i\omega}|  = |1 - D_z^{kk}|^{-1}$ satisfies  (\ref{eq:Paley-Wiener}).
The theorem above can be applied and the Lemma follows.

\vskip .3cm
\noindent
{\bf Main steps for kernel estimation.}
The different steps for the final kernel estimator, in the fully symmetric case,
 can be summarized as follows
\begin{itemize}
\item Set $\Delta$ small enough (see Section 4.1) and fix $h = \Delta$,
\item Estimate the unconditional intensity $\meanlambda$,
\item Estimate the auto-covariance operator $v_t^{(h)}$ and compute its Fourier transform $\widehat v_{i\omega}^{(h)}$,
\item Compute   $(\mathbb{I} + \widehat \Psi^*_{i\omega})(\mathbb{I} + \widehat{\Psi}^*_{i\omega})^\dagger$ using Eq. (\ref{eq:SignaturePlotExpressionFourier1}). Diagonalize it and compute the matrix $E_{i\omega}$ defined by Eq. (\ref{eq:Ez}),
\item Compute the diagonal matrix $D_{i\omega}$ using (\ref{eq:lem}),
\item Go back to the initial basis and Inverse Fourier transform to get the estimation of $\phi_t$.
\end{itemize}

\subsection{Some particular cases}

\subsubsection{The one dimensional case $n = 1$}
\label{result1}

As we already pointed out, the hypothesis \textbf{H3} is always true in this case since all the functions are scalar functions. Thus, the phase determination problem is solved without adding any assumption apart from  \textbf{H1} and \textbf{H2}.
The kernel estimator simply consists in first computing
\begin{equation}
   |1 + \widehat \Psi_{i\omega}|^2 = \frac{ \widehat v_{i\omega}^{(h)}}{\meanlambda \widehat g^{(h)}_{i\omega}},
\end{equation}
and then inverse Fourier transform of
\begin{equation}
\widehat \phi_{i\omega} = 1-e^{-\log| 1 + \widehat \Psi_{i\omega}| +i H(\log| 1 + \widehat \Psi_{i\omega}|)}.
\end{equation}

\subsubsection{The bisymmetric 2-dimensional case $n = 2$}
\label{result2}
In the two-dimensional case, the hypothesis \textbf{H3} is satisfied in the particular case
\begin{equation}
\mu^1 = \mu^2,
\end{equation}
and the kernel $\phi_t$ is {\em bisymmetric}, i.e.,  has the form
\begin{equation}
\label{sym}
\phi_t =
 \left(
\begin{array}{cc}
 \phi^{11}_t & \phi^{12}_t \\
 \phi^{12}_t & \phi^{11}_t \\
\end{array}
\right)
\end{equation}

The matrix of the Laplace transform $\widehat \phi_z$ can be indeed decomposed
as follows:
\begin{equation*}
    \widehat \phi_z= U^{\dagger}
\left(
\begin{array}{cc}
 - \widehat \phi^{11}_z +\widehat \phi^{12}_z & 0 \\
 0 &  \widehat \phi^{11}_z +\widehat \phi^{12}_z
\end{array}
\right) U
\end{equation*}
where $U=\frac{1}{\sqrt{2}}\left(
\begin{array}{cc}
 -1 & 1 \\
 1 & 1
\end{array}
\right)$.

Diagonalizing and identifying the diagonal coefficients in both hand sides of (\ref{eq:Ez}) leads to

\begin{eqnarray}
\frac{\widehat v^{(h)11}_{i\omega}+\widehat v^{(h)12}_{i\omega}}{\meanlambda\widehat g^{(h)}_{i\omega}} &=& |1 + \widehat \Psi_{i\omega}^{11}+\widehat \Psi_{i\omega}^{12}|^2 \label{eq:OnePlusFiMagnitudeN2_1}\\
\frac{\widehat v^{(h)11}_{i\omega}-\widehat v^{(h)12}_{i\omega}}{\meanlambda\widehat g^{(h)}_{i\omega}} &=& |1 + \widehat \Psi_{i\omega}^{11}-\widehat \Psi_{i\omega}^{12}|^2 \label{eq:OnePlusFiMagnitudeN2_2}
\end{eqnarray}

Applying the same method used in the case $n=1$ to these last two equations,  we get an estimate of $\widehat \Psi$. Finally we apply
\begin{equation}\label{eq:FiFuncPsiBySym}
    \widehat \phi_z = (\mathbb{I} + \widehat \Psi_z)^{-1}\widehat \Psi_z
\end{equation}
giving $\widehat \phi$. Applying the inverse transform to $\widehat \phi$ we finally get $\phi$.

In the following section we illustrate these results and our methods on numerical
simulations of 1D and 2D Hawkes processes.

\section{Numerical illustrations \label{Sec:Num}}

Let us discuss some examples illustrating the estimation method as defined previously using
simulated Hawkes processes. All the simulations have been performed with the thinning algorithm
described in \cite{ogata_lewis_1981}.

We estimate $v^{(h)}_\lag$  from the realization of a Hawkes process $(X_t)$, $t \geq 0$ (1D or 2D). We then strictly follow the method described in previous section in order to estimate the decay kernel $\phi$.
$v^{(h)}_{\lag}$ is sampled at rate $\Delta$ (i.e., $\lag = n\Delta$) up to a maximum lag $\lag_{max}$. As explained in Section \ref{estimation_principle},  in order to avoid problems related to the zeros of $\widehat g^{(h)}$, we naturally choose $h = \Delta$. From a practical point of view, $\Delta$ has to be chosen small enough in order to avoid Fourier aliasing.

\subsection{The case $n=1$}\label{sec:Est:sub:N1:NumEst}

\subsubsection{Exponential kernel}
We first simulate a one dimensional Hawkes process with an exponential kernel:
\begin{equation}
\label{eq:FiExponentialDecay}
    \phi_t ~~=~~\alpha e^{-\beta t} \mathds{1}_{t\geq 0}
\end{equation}
where we choose ($\mu=1$, $\alpha=1$, $\beta=4$). This gives $\widehat \phi_0 = 1/4$, $\meanlambda= 4/3$. The simulated sample contains $130000$  jumps (it is approximately $T \simeq 10^5$ seconds long). In figure \ref{fig:Hilbert-Estimation}(a) we have reported the estimated kernel function $\phi_t$ (circles) on top of the true kernel (solid line). We see that the estimated kernel function is, up to some noise, very close to the real kernel
\begin{figure}[h]
\centering
\hspace{\myspace}
\hspace{\myspaceside}
\resizebox{\myfactside\textwidth}{!}{
\includegraphics{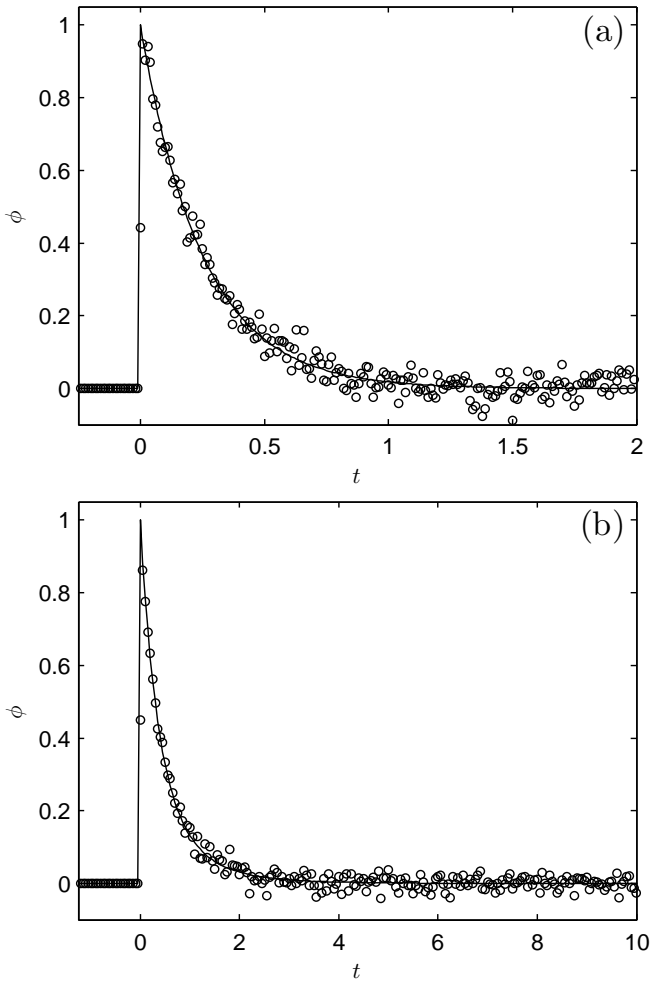}
}
\caption{Non parametric estimation
 of a one dimensional Hawkes process using method described in Section \ref{result1} from a unique realization with  130000 jumps. Estimated ($\circ$) and analytical kernel (solid line) are shown. (a) Case of the exponential decay kernel $\phi$ (Eq. (\ref{eq:FiExponentialDecay})) with $\alpha=1$, $\beta=4$. We used $\Delta=0.01$ and $\lag_{\max}=2$. (b) Case of the power-law decay kernel $\phi$ (Eq. (\eqref{eq:FiPowerLawDecay})), with $\alpha=32$, $\beta=-5$ and $\gamma=2$. We used $\Delta=0.05$ and $\lag_{\max}=20$.}
 \label{fig:Hilbert-Estimation}
\end{figure}

\subsubsection{Power Law Decay}
We now consider a power-law decaying kernel defined as:
\begin{equation}\label{eq:FiPowerLawDecay}
    \phi_t~~=~~\alpha (t+\gamma)^{\beta} \mathds{1}_{t\geq 0}
\end{equation}
with $\beta < -1$. 
In this case we have $\widehat \phi_0=-\frac{\alpha}{\beta+1}\gamma^{\beta+1}$. We choose $\alpha=32$ and $\beta=-5$ and $\gamma=2$ making $\widehat \phi_0=0.5<1$ and $\meanlambda=2$ again with 130000 jumps ($T \simeq 65000$ secs). The kernel estimated on a single sample is reported in \ref{fig:Hilbert-Estimation}(b). Once again, one can see that, up to an additive noise, the estimated kernel fits well the real one. Let us point out that, since the decay is much slower than in the previous exponential case, we chose the maximum lag $\lag_{max}$ to be ten times as much.

\subsubsection{Error Analysis}
\label{erroranalysis}
Let us briefly discuss some issues related to the errors associated with our kernel estimates. In ref. \cite{bacry_scaling_2011}, a central limit theorem has been proved that shows that, asymptotically, the errors of the empirical covariance function estimates are normally distributed with a variance that decreases as $T^{-1}$ (or $N^{-1}$ for $\Delta$ fixed). One thus expects the same kind of results in the estimates of the components of $\phi$.

Let us define the $L^2$ estimation error as:
\begin{equation}\label{eq:Est_Error_2}
    e^2=\sum_{k=1}^{\frac{\lag_{\max}}{\Delta}}|\phi_{k\Delta} - \phi^{(e)}_{k\Delta}|^2
\end{equation}
where $\phi^{(e)}$ is the estimated kernel. In Fig. \ref{fig:ConvergenceSpeed}(a), we have reported $e^2$ as a function of the sample length $T$ for the same realization as the one used in Fig. \ref{fig:Hilbert-Estimation}(a). As expected, one observes a behavior very close to $T^{-1}$.
\begin{figure}[h]
\centering
\hspace{\myspaceside}
\resizebox{\myfactside\textwidth}{!}{
       \includegraphics{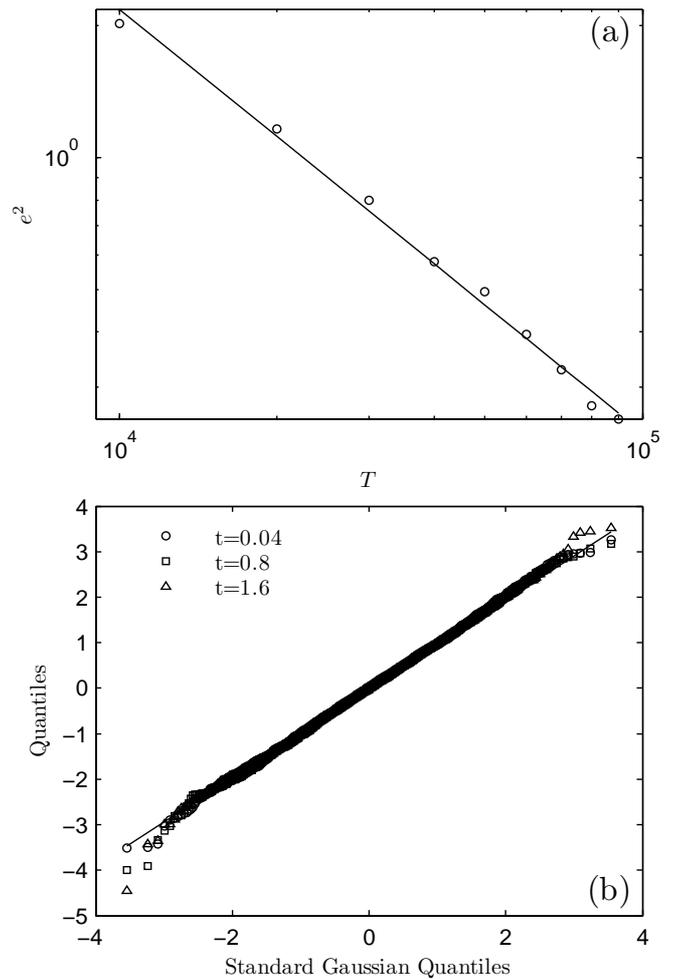}
}
        \caption{Error analysis of a one dimensional Hawkes process with the same exponential kernel as in Fig. \ref{fig:Hilbert-Estimation}(a) with $\Delta = 0.01$ and $\lag_{max = 2}$.
        (a) Mean square error as a function of the sample length $T$ in log-log coordinates. The solid-line corresponds to the curve the expected $T^{-1}$ behavior. (b) qq-plots of the Normal probability distribution function (pdf) versus empirical error pdf for $\phi_t$ for three different values of $t$. Here $T=10^5$ seconds is fixed. The empirical error pdf's have been normalized to have the same variance (the variance increases as $t$ goes to 0).
}
\label{fig:ConvergenceSpeed}
\end{figure}
Let us now look at the error between the analytical $\phi$ and the estimated one for a fixed $t>0$. We find that, for each $t$, the series of errors are centered gaussian and mainly uncorrelated. Their variance increases as $t$ decreases to 0.  The normality of the observed errors is illustrated in Fig. \ref{fig:ConvergenceSpeed}(b) where we report, for 3 different values of $t$, the qq-plots of the empirical error pdf (with standardized variance) versus the normal pdf.

Let us point out that the estimation error depends, a priori, on the sampling rate $\Delta$. Numerical simulations show that there is an optimal choice for the sampling rate parameter $\Delta$. Indeed, we found that
the estimation error is minimum for a finite value $\Delta = \Delta^*$ which is neither large nor close to 0. For instance, for the process used in Fig. \ref{fig:Hilbert-Estimation}(a), we found $\Delta^* \simeq 0.15$. The fact that there exists such a minimum is not that surprising. On the one hand, if $\Delta$ is too large the error is dominated by Fourier aliasing. On the other hand, if $\Delta$ is too small, estimating $v^{(\Delta)}_{\lag}$ for a fixed $\lag>0$ corresponds to estimating the correlation between the increments at scale $\Delta$ of two point processes. Such an estimation is well known to converge to 0 when the time-scale $\Delta$ goes to 0. In the Finance literature (see, e.g., \cite{bacry_modeling_2011}),  this is known as the {\em Epps effect}. Of course, the optimal value $\Delta^*$ is not known a priori. However, it is natural to think that it is inversely  proportional to $\meanlambda$. On practical situations we advocate to start choosing $\Delta$ of the order of $0.1\meanlambda$ and then play around this value.

\subsection{The bisymmetric 2-dimensional case $n = 2$ with an exponential kernel}
\label{ssec:case-n=2}
Let us now consider a 2-dimensional bisymmetric Hawkes process introduced in Section \ref{result2}. Both the diagonal term $\phi^{(d)}$ and the anti-diagonal term $\phi^{(a)}$ (cf Eq. \eqref{sym}) have an exponential form
\begin{align}\label{eq:Fi2DimSimu}
\phi_{t}^{(d)} =   \phi^{11}_t &= \phi^{22}_t = \alpha_de^{-\beta_d t} \mathds{1}_{t\geq 0} \\
\nonumber
\phi_{t}^{(a)} =   \phi^{12}_t &= \phi^{21}_t = \alpha_ae^{-\beta_a t} \mathds{1}_{t\geq 0}.
\end{align}
We use the following parameters for the simulations, $\alpha_d=0.5$, $\beta_d=8$, $\alpha_a=1$, $\beta_a=4$, and $\mu^1 = \mu^2 = 1$ (consequently $\meanlambda = 1.45$). We simulate about $60000$ jumps for each of the two components of the Hawkes process. The Figure \ref{fig:Hilbert-EstSimu-Expon-dim-2} below shows the theoretical (solid lines) and estimated ($\circ$)versions of $\phi_d$ (in Fig.  \ref{fig:Hilbert-EstSimu-Expon-dim-2}(a) and $\phi_a$ in Fig.  \ref{fig:Hilbert-EstSimu-Expon-dim-2}(b)).
We see that, as in the 1D case, we get a reliable estimate of both kernels.
\begin{figure}[h]
    \centering
    \hspace{\myspaceside}
    \hspace{\myspace}
\resizebox{\myfactside\textwidth}{!}{
\includegraphics{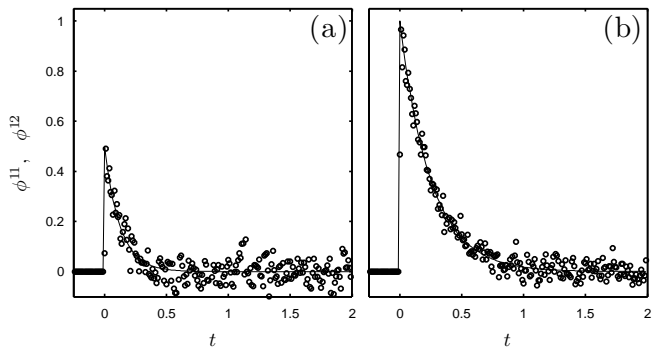}
}
        \caption{Non parametric estimation of the two dimensional Hawkes bisymmetric exponential
         kernel $\phi$ (Eq. (\ref{eq:Fi2DimSimu})) from a unique realization with  60000 jumps.
         We chose $\alpha_d=0.5$, $\beta_d=8$, $\alpha_a=1$, $\beta_a=4$, and $\mu^1 = \mu^2 = 1$ ($\meanlambda = 1.4545$)
         We used $\Delta=0.05$ and $\tau_{\max}=2$.  Estimated ($\circ$) and analytical kernel (solid line) are shown.
         (a) Estimation of the diagonal term $\phi_{t}^{(d)}$. (b) Estimation of the anti-diagonal term $\phi_{t}^{(a)}$}
         \label{fig:Hilbert-EstSimu-Expon-dim-2}
\end{figure}

\section{Application to high frequency market data \label{sec:Applica-high-frequen-market-data}}
As mentioned in the introduction, Hawkes point processes have found many applications and notably as models for high frequency financial data. Indeed, because of the discrete and correlated nature of trade and limit order arrival times, point processes are  natural  models of market dynamics at the microstructure level. One can mention for instance Ref. \cite{hewlett_clustering_2006} where buy and sell trades arrivals are represented by a bivariate Hawkes process with exponential kernels (see also \cite{large_measuring_2007}). Another approach, developed in \cite{bacry_modeling_2011}, consists in describing high frequency price dynamics as the difference between two coupled Hawkes processes representing respectively up and down discrete price variations. The authors emphasized that such model allows one to account for the main stylized facts characterizing the observed noise microstructure, namely the signature plot and the Epps effect (see end of Section \ref{model2} for a short remark on these effects).

We use \textbf{level 1} data (i.e, trades and best limit data) provided by QuantHouse Trading Solutions\footnote{http://www.quanthouse.com/} of  the 10-years Euro-Bund (Bund) and on the futures contracts on the Dax index. In this paper, we use $75$ days covering the period between 2009-06-01 and 2009-09-15. The most liquid maturity is always chosen. In order to minimize seasonal effect, every day, only the data between $9$ AM and $11$ AM (GMT) are kept during which the rate of incoming orders can be considered as stationary, moreover, we ignored the days with too little trades.
The data has millisecond accuracy and it has been treated in such a way that each market order is equivalent to exactly one trade
\footnote{When one market order hits several limit orders it results in several trades being reported.}.

\subsection{Estimation in the case of the one-dimensional model for Bund data}\label{sec:Applica-high-frequen-market-data:ssec:n-1}
\label{numresult1}
In the following we test our method on the process of incoming trade times (market orders) for the Bund Future. This is a 1-dimensional point process, consequently we use the estimator described in Section \ref{result1}.
Every single day, we compute an estimation of the function $v_{n\Delta}^{(h)}$ with $h=\Delta=0.1$ and $n\Delta<100$. These quantities are then averaged on all the 75 days and we perform the $\phi$ estimation using this averaged quantity.

The results are shown in the figure \ref{fig:Hilbert-Estimation-FBund-Tmax-100-dela_step-0.1}.
The so-obtained estimation of $\phi$ clearly displays  a power law decay (Fig. \ref{fig:Hilbert-Estimation-FBund-Tmax-100-dela_step-0.1}(b)).
A fit with the function $\alpha t^{\beta}$ gives  $\alpha \simeq 0.1$ , and $ \beta \simeq  -1$\footnote{Strictly speaking a decay kernel of the form $\alpha t^{-1}$ is not admissible for our model. Indeed, its integral diverges both at $t=0$ and at infinity. So it is admissible as long as we consider that this behavior has two cut-off, for small and large $t$. This hypothesis will be implicitly made in the following.}.

Let us point out that we studied the stability of this estimation as the value for $\Delta$ is changed. Table \ref{tab:FitResultsFGBLOneDim}, clearly shows that, as $\Delta$ decreases,  there is a pretty large discrepancy on the estimated values of both both $\alpha$ and $\beta$.
Though these values seem to stabilize when $\Delta$ reaches 0.1 which could indicate that this choice is not far from the optimal value $\Delta^*$ (see Section \ref{erroranalysis})
\begin{figure}[h]
\centering
\hspace{\myspace}
\hspace{\myspaceside}
\resizebox{\myfactside\textwidth}{!}{
       \includegraphics{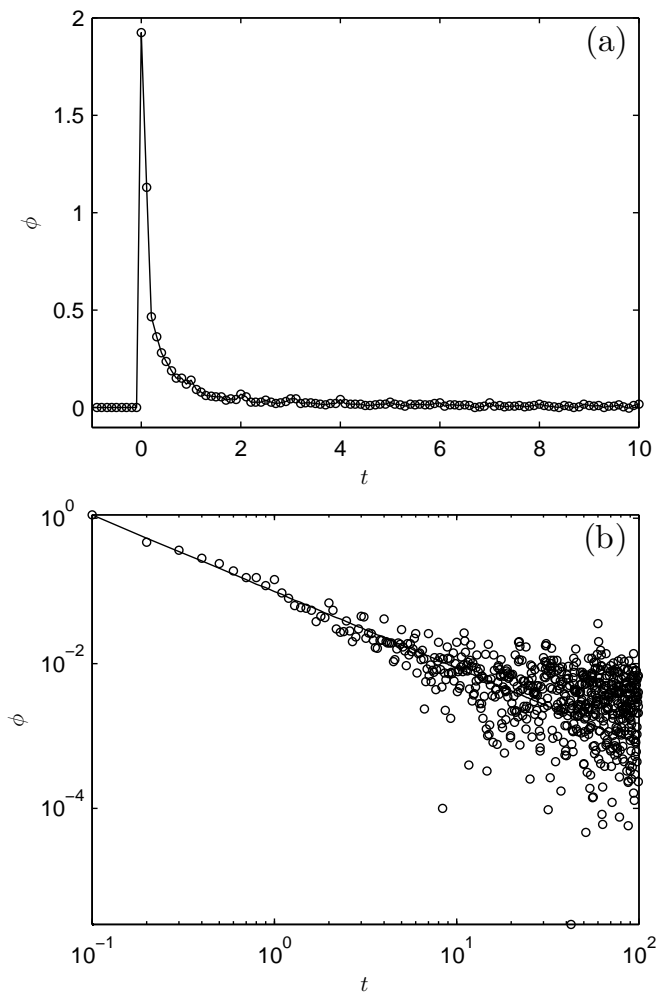}
}
        \caption{1-dimensional non parametric estimation of $\phi$ for the point process of incoming market orders of the Bund futures.
        We used $\Delta = h = 0.1$ and $\lag_{\max}=100$. (a) The estimation clearly displays a slow decay. (b) Log-Log plot the figure (a) above reveals that, in good approximation, the kernel can be considered as decaying as $t^{-\beta}$. A power-law least-square fit
        (solid line) provides the exponent $\beta = -1.05$.}
        \label{fig:Hilbert-Estimation-FBund-Tmax-100-dela_step-0.1}
\end{figure}

\begin{table}[h]
\centering
\begin{tabular}{ccc}
\toprule
$\Delta$ & $\alpha$ & $\beta$\\ [0.5ex]
\otoprule
1 & 0.146276 & -1.41515 \\
0.5 & 0.117503 & -1.30292 \\
0.1 & 0.098624 & -1.05329 \\
0.05 & 0.092975 & -1.03596 \\
0.01 & 0.089227 & -0.99899 \\\bottomrule
\end{tabular}
\caption{Results of the Power Law fit $\alpha t^\beta$, applied to the estimated Hawkes kernel for the rate of incoming market orders of the Bund Futures. We show the results for different values of the parameter $\Delta$.
As $\Delta$ decreases,  the estimations for both $\alpha$ and $\beta$ stabilize and seem to indicate that the choice of $\Delta = 0.1$ is not far from the optimal value $\Delta^*$.
}\label{tab:FitResultsFGBLOneDim}

\end{table}

\subsection{Estimation in the case of the two-dimensional model}\label{sec:Applica-high-frequen-market-data:ssec:n-2}

\subsubsection{The two dimensional model}
\label{model2}
In this Section, we apply our estimation framework to
the process that was initially introduced in \cite{bacry_modeling_2011} for modeling the changes in the mid-price of a given asset.
Let $X_t$ be the mid-price, we decompose it as the sum of the cumulative positive jumps $N^+_t$ and the cumulative negative jumps\footnote{strictly speaking, the jumps have not always the same amplitude. Though for the Bund, most of them are 1-tick large it is not the case for the DAX data. If it is not the case $N^+_t$ (resp. $N^-_t$) just represent the point process (with jumps always equal to 1) with the same arrival time as the upward (resp. downward) jumps of $X_t$} $N^-_t$
\begin{equation*}
    X_t=N^+_t - N^-_t.
\end{equation*}
We look at the two dimensional point process
\begin{equation*}
    N_t=\left(
           \begin{array}{c}
             N^+_t \\
             N^-_t \\
           \end{array}
         \right)
\end{equation*}
As advocated in Ref. \cite{bacry_modeling_2011}, 2 dimensional Hawkes processes are suited to model the so-defined $N_t$.
A simple bisymmetric exponential kernel  (with a null diagonal term) is able to reproduce remarkably the so-called {\em signature plot}.
Moreover, in the case of 2 assets, each of them modeled by a 2-dimensional Hawkes and their interaction modeled by a simple symmetric cross term (leading to a 4-dimensional Hawkes process), these models have also been able to reproduce the so-called {\em Epps effect}.

The Signature plot and the Epps effect are two of the main stylized facts of financial time-series at a microstructure level.
It is not the goal of this paper to go into more details about them, however, it is interesting to point out that the quantity $v_\lag^{(h)}$ on which our  estimation procedure is based includes both the Signature plots (basically corresponding to the diagonal terms of the
$h \rightarrow v_0^{(h)}$ matrix function) and the Epps effect (basically corresponding to the non diagonal terms of the same function).

\subsubsection{Validating hypothesis \textbf{H3}}
In order to apply our estimation framework in dimension 2, we first need to check the hypothesis \textbf{H3}. This hypothesis is two-folded.

\begin{itemize}
\item  First of all, on Fig. \ref{fig:Lambdas-FGBL-Tmax-20-dela_step-0.1-dim-2}, we show for each single day an estimation of $\Lambda^1=\Lambda^+$ versus
 $\Lambda^2=\Lambda^-$. The plot shows that $\Lambda_1 \simeq \Lambda_2$ with a very small variation and therefore that we are within our assumption that $\Sigma=\meanlambda\mathbb{I}$ with $\meanlambda = \Lambda^1 = \Lambda^2$.
\item Secondly, we need to show that the kernel matrix diagonalizes in a basis that is constant. Actually, Fig. \ref{fig:SACF-FGBL-Tmax-20-dela_step-0.1-dim-2} shows that $v_{\tau}^{(h)}$ (and consequently the kernel) is, in good approximation bisymmetric which implies that it diagonalizes on a constant basis.
Indeed we see that $v_{\tau}^{(h),11}$ ($\circ$, in Fig. \ref{fig:SACF-FGBL-Tmax-20-dela_step-0.1-dim-2}(a)) is very close to $v_{\tau}^{(h),22}$ ($\triangle$, in Fig. \ref{fig:SACF-FGBL-Tmax-20-dela_step-0.1-dim-2}(a))
 for all $\lag>0$. In the same way, we see that $v_{\tau}^{(h),12}$ ($\circ$, in Fig. \ref{fig:SACF-FGBL-Tmax-20-dela_step-0.1-dim-2}(b)) is also (up to some estimation noise) equal to $v_{\tau}^{(h),21}$ ($\triangle$, in Fig. \ref{fig:SACF-FGBL-Tmax-20-dela_step-0.1-dim-2}(b))
 for all $\tau>0$ (same results are obtained for $\tau<0$).
 \end{itemize}

Let us point out that the choice of the midpoint price series is motivated by two factors. First, if we choose the series of traded prices, we have, because of a non-zero spread value, an important ``spurious" bouncing effect (oscillation between best bid and best ask prices) that is very hard to capture by our modeling approach (see \cite{bacry_modeling_2011}) and we get negative decay functions on the diagonal of $\phi$: $\phi^{11}$ and $\phi^{22}$ (see Section \ref{estim2} for a longer discussion on the influence of the bouncing effect on the kernel estimation). Second, if we choose the series of the last traded prices of buy orders only as in \cite{bacry_modeling_2011}, the bouncing artifact disappears, however the symmetry of $\Lambda^+$ and $\Lambda^-$ is naturally  no longer verified. The series of midpoint prices has the advantage of having a reduced bouncing effect and of presenting identically distributed $N^+$ and $N^-$ processes\footnote{
Still, when actually going through the estimation process, for stability reasons, we chose $\meanlambda$ to be the average of the estimated $\Lambda^+$ and $\Lambda^-$ and, similarly, $v_{\tau}^{(h),11}$ and $v_{\tau}^{(h),22}$ have been averaged as well as
$v_{\tau}^{(h),12}$ and $v_{\tau}^{(h),21}$
}.

\begin{figure}[h]
    \centering
\hspace{\myspace}
\resizebox{\myfactside\textwidth}{!}{
 \includegraphics{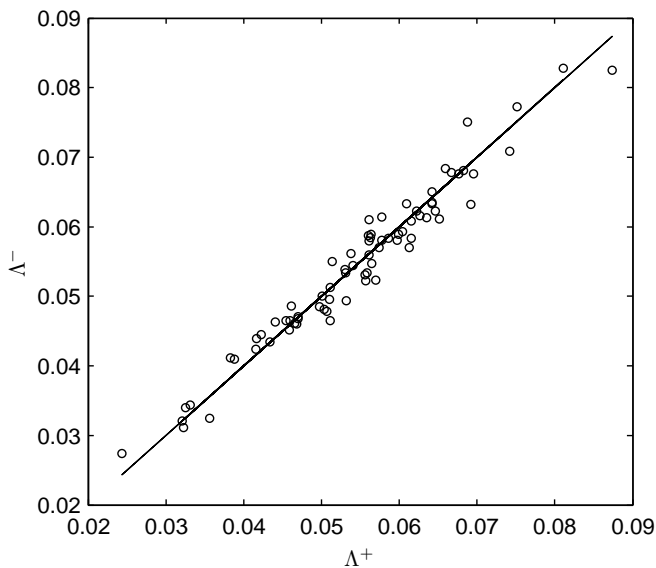}
 }
        \caption{Estimated $\Lambda^1=\Lambda^+$ versus $\Lambda^2=\Lambda^-$
        for every single day (75 total) in the 2-dimensional model for mid-price of the Bund futures. Each dot represents a single day. The solid line corresponds to $\Lambda^+=\Lambda^-$. We see that, in good approximation
        $\Lambda^+\simeq \Lambda^-$ validating the first part of hypothesis \textbf{H3}.}
        \label{fig:Lambdas-FGBL-Tmax-20-dela_step-0.1-dim-2}
\end{figure}

\begin{figure}[h]
    \centering
\hspace{\myspaceside}
\resizebox{\myfactside\textwidth}{!}{
        \includegraphics{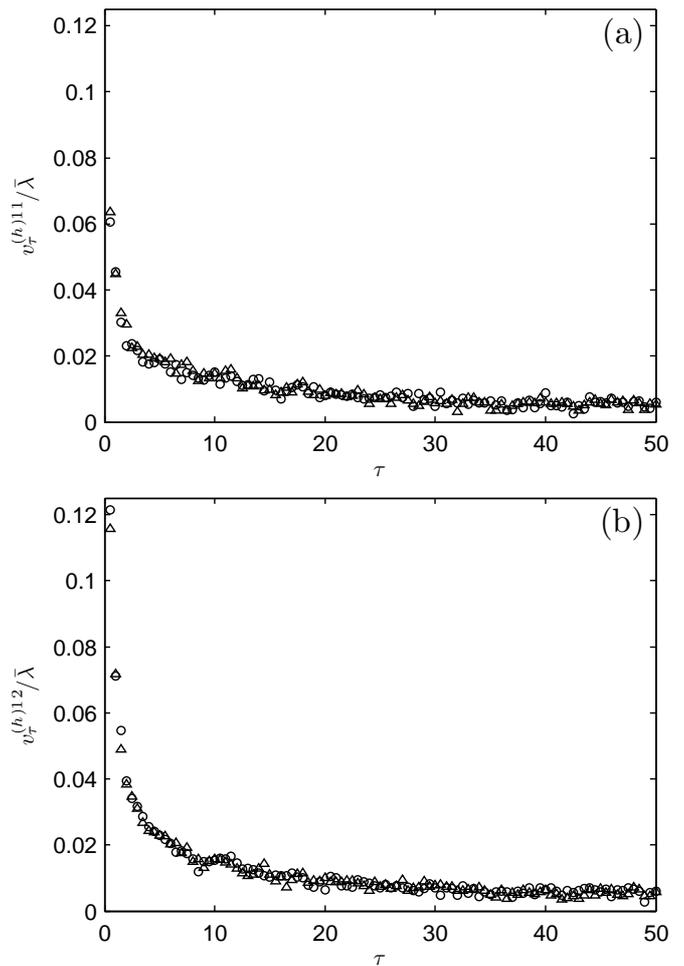}
        }
        \caption{Estimation of $v_{\lag}^{(h)}$ for  $\lag >0$ on 75 days (from 9am to 11am) in the framework of the 2-dimensional model for mid-price of the Bund futures. We see that the matrix $v_{\lag}^{(h)}$ is, in good approximation, bisymmetric. Same results would be obtained for $\lag<0$. That completes (with Fig. \ref{fig:Lambdas-FGBL-Tmax-20-dela_step-0.1-dim-2}) the validation of hypothesis \textbf{H3}.
        (a)  $v_{\tau}^{(h),11}$ ($\circ$)  and  $v_{\tau}^{(h),22}$ ($\triangle$).
(b)  $v_{\tau}^{(h),12}$ ($\circ$)  and  $v_{\tau}^{(h),21}$ ($\triangle$).
}\label{fig:SACF-FGBL-Tmax-20-dela_step-0.1-dim-2}
\end{figure}

\subsubsection{Kernel estimation}
\label{estim2}
We use the framework developed in Section \ref{result2} in order to estimate $\phi$.

The result of the non parametric estimation is shown in figure \ref{fig:Hilbert-fi11Andfi12-Estimation-FGBL-Tmax-20-dela_step-0.1-dim-2}. We immediately notice that the diagonal term $\phi^{(d)} = \phi^{11}$ is an order of magnitude smaller than the non diagonal one $\phi^{(a)} =\phi^{12}$ and can be considered as being zero. This means that $N^+$ and $N^-$ are not self exciting and exclusively mutually exciting.
The log-log plot in Fig.  \ref{fig:Hilbert-fi12-Estimation-FGBL-Tmax-20-dela_step-0.1-dim-2} of the anti-diagonal term $\phi^{(a)} = \phi^{12}$ displays a power-law behavior $\alpha t^\beta$ with
$\alpha \simeq  0.095$ and $\beta \simeq  -0.99$ which is unsurprisingly close to the values we found for the 1-dimensional model in
Section \ref{numresult1}.

\begin{figure}[h]
    \centering
\hspace{\myspaceside}
\hspace{\myspace}
\resizebox{\myfactside\textwidth}{!}{
    \includegraphics{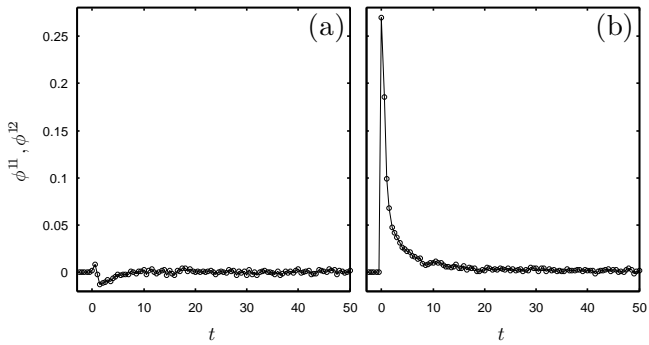}
}
\caption{Non parametric estimation of the bisymmetric kernel $\phi$ for the mid-price of the Bund futures using 75 days (between 9am to 11am). We used $\Delta = h = 0.5$ and $\lag_{\max}=200$. The diagonal term $\phi^{(d)}=\phi^{11}$ appears to be negligible compared to the anti-diagonal term $\phi^{(a)}=\phi^{12}$ confirming that the $N_+$ and $N_-$ are mutually but not self exciting (a) Estimation of the diagonal term $\phi_{t}^{(d)}$. (b) Estimation of the anti-diagonal term $\phi_{t}^{(a)}$.}\label{fig:Hilbert-fi11Andfi12-Estimation-FGBL-Tmax-20-dela_step-0.1-dim-2}
\end{figure}

\begin{figure}[h]
    \centering
\hspace{\myspaceside}
\resizebox{\myfactside\textwidth}{!}{
        \includegraphics{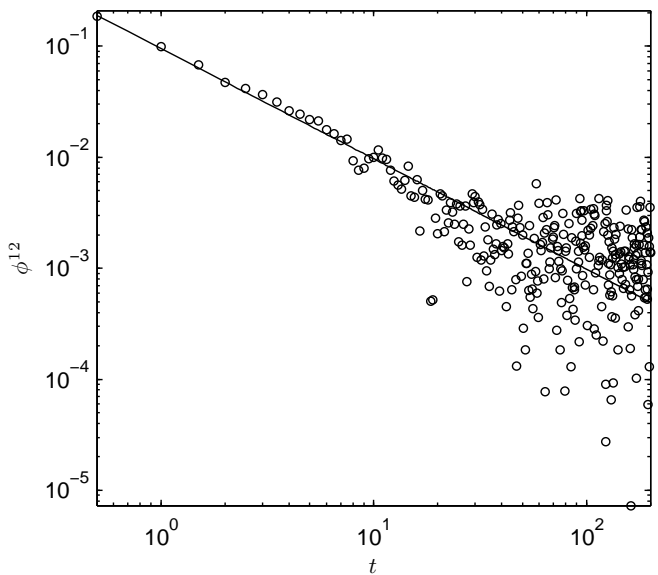}
        }
        \caption{Log-Log plot of the anti-diagonal kernel $\phi^{(a)} = \phi^{12}$ (estimated for the Bund futures ) shown on Fig. \ref{fig:Hilbert-fi11Andfi12-Estimation-FGBL-Tmax-20-dela_step-0.1-dim-2}(b).
        It shows that, in good approximation, this term can be considered as decaying as $t^{-\beta}$
        A power law fit (solid line) gives the exponent $\beta \simeq -1$.  }\label{fig:Hilbert-fi12-Estimation-FGBL-Tmax-20-dela_step-0.1-dim-2}

\end{figure}

Figure \ref{fig:Hilbert-mid-fi11Andfi12-Estimation-FDax-Tmax-200-dela_step-0.5-dim-2} show the estimation of the bisymmetric kernel on Dax Futures time-series. Here, the diagonal kernel is of the same order (it can even be greater) than the anti-diagonal term. This property, that was not observed for the Bund, can be interpreted using tick size considerations. Indeed, the tick size of Dax Futures is well known to be very small in the sense that the agents trading on this asset only care for moves that are greater than 1 tick (or equivalently for successive moves of 1 tick in the same direction). This translates into the fact that, at the scale of 1-tick, there is hardly no bouncing effect (negative autocorrelation of the returns). On the contrary, on the Bund, the tick size is well known to be "too big", i.e., when the price goes up by 1-tick, most of the time the next move is a downward move, i.e., the bouncing effect is very strong. In the framework of our model, it is clear that the strongest the bouncing effect the smaller the anti-diagonal term. Actually, looking carefully at the Fig. \ref{fig:Hilbert-fi11Andfi12-Estimation-FGBL-Tmax-20-dela_step-0.1-dim-2}(a), the bouncing effect on the Bund is so strong that, though much smaller than the anti-diagonal kernel, the estimated  diagonal kernel seems to become significantly negative (compared to the estimation noise)  which  defies  the assumption of our model. Indeed, if $\phi^{11}$ becomes negative, there is a priori, in theory,  no guarantee that $\lambda_t$ remains positive, and a negative rate of arrival is unacceptable. Of course, in practice, since these negative values are much smaller than the anti-diagonal kernel, the probabilities for $\lambda_t$ to become negative are clearly negligible.
The bouncing effect is most apparent on the series of trade prices (bouncing between best ask and best bid prices) and while it is heavily dampened when we use the series of midpoint prices,  we see that it is still too strong to be fully captured by our model (see the discussions in \cite{bacry_modeling_2011} and in \cite{dayri_2011}).

\begin{figure}[h]
\centering
\hspace{\myspaceside}
\hspace{\myspace}
\resizebox{\myfactside\textwidth}{!}{
        \includegraphics{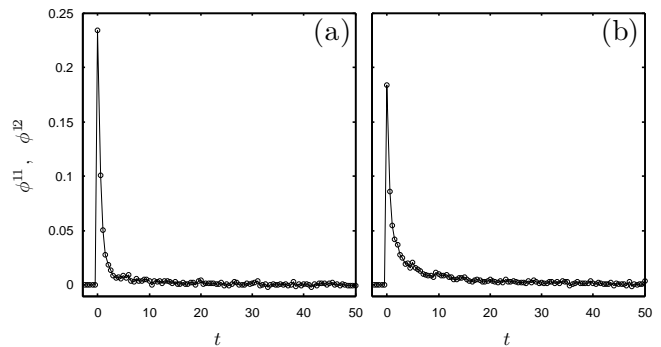}
        }
        \caption{Non parametric estimation of the bisymmetric kernel $\phi$ for the mid-price of the Dax futures using 75 days (between 9am to 11am). We used $\Delta = h = 0.5$ and $\lag_{\max}=200$. Contrarily to the Bund kernel (Fig. \ref{fig:Hilbert-fi11Andfi12-Estimation-FGBL-Tmax-20-dela_step-0.1-dim-2}), the diagonal term $\phi^{(d)}=\phi^{11}$ appears to be of the same order as the anti-diagonal term $\phi^{(a)}=\phi^{12}$. This is due to the tick size of the Dax which is much smaller than the tick size of the Bund. (a) Estimation of the diagonal term $\phi_{t}^{(d)}$ $\beta \simeq -1.2$ . (b) Estimation of the anti-diagonal term $\phi_{t}^{(a)}$ $\beta \simeq -0.8$.}\label{fig:Hilbert-mid-fi11Andfi12-Estimation-FDax-Tmax-200-dela_step-0.5-dim-2}
\end{figure}

\section{Conclusion and prospects}
In this paper we have introduced a non parametric method to estimate the shape of the self-exciting
kernels for symmetric Hawkes processes. Our method can be implemented very easily since it relies on the computation
of the empirical covariance matrix and mainly uses Fourier transforms.
As illustrated on specific numerical examples (1D and 2D),
it provides reliable results for series of about $10^5$ events long.
This method can be very helpful to get a precise idea
of the kernel functional shape before proceeding to a classical parametric (e.g. maximum likelihood) estimation. In a future work,
we will consider a natural extension of the method in order to
account for random marks associated with each events as e.g. in the ETAS
model for earthquakes \cite{ogata_seismicity_1999,sornette}.

As far as financial time series are concerned, we have shown that,
for both Bund and Dax Futures high frequency data, the Hawkes kernels are slowly (power law) decaying.
Even if this observation has to be confirmed by further studies,
it is noteworthy that the values of the kernel power-law exponent we found ($\beta \simeq -1$) corresponds
to the onset of stationarity.
Our findings can be of great interest since they allow one to finely
describe, at high frequency, the market activity as
a long-memory self-exciting process. This can notably
allow one to improve the results of Ref. \cite{bacry_modeling_2011}
where the authors restrict themselves to exponential kernels to reproduce
noise microstructure main features using Hawkes processes.
On a more general ground, one can hope that our findings will be helpful
to bridge the gap between a theory of price variations at the microstructure
level and the standard coarse scale models with long range correlated volatility.

\bibliographystyle{plain} 
\bibliography{BDM_EPJB}

\end{document}